\documentclass[a4paper, 10pt, conference]{ieeeconf}      

\IEEEoverridecommandlockouts                              

\usepackage{graphics} 
\usepackage{epsfig} 
\usepackage{mathptmx} 
\usepackage{times} 
\usepackage{amsmath} 
\usepackage{amssymb}  
\usepackage{amsfonts}
\usepackage{graphicx}
\usepackage{caption}
\usepackage{subcaption}
\usepackage{algorithmicx}
\usepackage{algorithm}
\usepackage{algpseudocode}
\usepackage{pseudocode}
\usepackage{verbatim}
\usepackage{mathtools}
\usepackage{fullpage}
\usepackage{url}
\usepackage{fancybox}

\usepackage{amsthm}

\newtheorem{theorem}{Theorem}[section]

\theoremstyle{definition}
\newtheorem{definition}{Definition}[section]

\newcommand{\norm}[1]{||#1||}
\newcommand{\onenorm}[1]{||#1||_1}

\newcommand{\R}{\mathbb{R}}

\newcommand{\hkpr}{\rho_{t,f}}

\newcommand{\hkapprox}{\hat{\rho}_{t,f}}

\renewcommand{\L}{\mathcal{L}}

\renewcommand{\H}{\mathcal{H}}

\newcommand{\hkpralg}{\texttt{ApproxHKPR}}
\newcommand{\hkpralgparams}{\texttt{ApproxHKPR($G,t,f,\epsilon$)}}
\newcommand{\rparam}{\frac{16}{\epsilon^3}\log n}
\newcommand{\kparam}{\frac{\log(\epsilon^{-1})}{\log\log(\epsilon^{-1})}}
\newcommand{\hkprcomplexity}{O\big(\frac{\log(\epsilon^{-1})\log n}{\epsilon^3\log\log(\epsilon^{-1})}\big)}
\newcommand{\consensusalg}{\textsc{AvgConsensus}}

\newcommand{\greensalg}{\texttt{GreensSolver}}
\newcommand{\greenscomplexity}{O\Big(\frac{ (\log s)^2(\log(\epsilon^{-1}))^2 }{ \epsilon^5\log\log(\epsilon^{-1}) }\Big)}
\newcommand{\lfconsensusalg}{\textsc{LFConsensus}}

\newcommand{\rparamsolver}{\frac{\log s + \log(\epsilon^{-1})}{\epsilon^2}}
\newcommand{\tparamsolver}{s^3\log(1/\epsilon)}

\title{\LARGE \bf
Finding Consensus in Multi-Agent Networks Using Heat Kernel Pagerank
}

\author{Olivia Simpson$^{1}$ and Fan Chung$^{2}$
\thanks{$^{1}$Dept. of Computer Science,
        University of California, San Diego,
        La Jolla, CA
        {\tt\small osimpson@ucsd.edu}}%
\thanks{$^{2}$Dept. of Mathematics,
        University of California, San Diego,
        La Jolla, CA
        {\tt\small fan@ucsd.edu}}%
}

\begin{document}

\maketitle
\thispagestyle{empty}
\pagestyle{empty}

\begin{abstract}

We present a new and efficient algorithm for determining a consensus value for a network of agents.
Different from existing algorithms, our algorithm evaluates the consensus value for very large networks using heat kernel pagerank.
We consider two frameworks for the consensus problem, a weighted average consensus among all agents, and consensus in a leader-following formation.
Using a heat kernel pagerank approximation, we give consensus algorithms that run in time sublinear in the size of the network, and provide quantitative analysis of the tradeoff between performance guarantees and error estimates.

\end{abstract}

\section{Introduction}
The problem of consensus among multi-agent systems has wide applications in situations where members of a distributed network must agree.  
For example, the communication, feedback, and decision-making between distinct unmanned aerial vehicles (UAVs)~\cite{om:graphridigity:02,fm:cooperationvehicle:04} is closely related to the consensus problem.
In addition to UAVs, the distributed coordination of networks has important implications in cooperative control of distributed sensor networks~\cite{cb:dynamicalsystems:05}, flocking and swarming behavior~\cite{tt:flocks:98}, and communication congestion control~\cite{pdl:congestion:01}.
Further, they form the foundation of the field of distributed computing~\cite{lynch:distributed:96}.
The consensus problem is studied in~\cite{om:consensusdynamic:03}, and several variations and extensions are examined by~\cite{om:switchingtop:04,nc:leaderfollower:10,mtll:dynamicthermal:11,cht:highorder:13}.  

We consider the classical model (see~\cite{om:consensusdynamic:03}) of agents with fixed, bidirectional communication channels and associated state.
State changes occur continuously, influenced by communication with neighbors.
A consensus algorithm is a continuous time protocol that specifies the information exchange between agents and provides a mechanism for systematically computing the consensus value, a unanimous state and an equilibrium of the system.
In this paper, we focus on an efficient method for approximating the state values in a network in which agents reach consensus by following a linear protocol.
We give algorithms for two different frameworks.
The first computes a global consensus value involving all the agents in the network and runs in time sublinear in the size of the network.
The second is a local algorithm to compute a consensus value for a subset of agents under external influence, and runs in time sublinear in the size of the specified subset.
In both, the consensus value returned is within an error bound of $O(\epsilon)$ for a given $0<\epsilon<1$.

Our algorithm involves solving a linear system by approximating the heat kernel pagerank of the network and relies on spectral analysis.
The tools we present are relevant to numerous graph problems, including partitioning and clustering algorithms~\cite{csz:spectralkway:94,shi2000normalized,njw:spectralcluster:02}, flow and diffusion modeling~\cite{raj2012network}, electrical network theory~\cite{ckmst:electric:10}, and regression on graphs~\cite{cai2007spectral}.

\subsection{Previous Work}
\label{sec:previouswork}
\paragraph{The consensus problem}
In~\cite{om:consensusdynamic:03}, Olfati-Saber and Murray design a linear protocol for agents to reach a consensus value which is an average of initial states.
They consider a network of agents as an undirected graph and use the Laplacian potential, defined in terms of the graph Laplacian (to be defined in Section~\ref{sec:prelimlaplacian}), as a measure of disagreement among nodes.
With this tool, they transform the the problem of reaching consensus to that of minimizing the Laplacian potential.

An alternate formulation given in~\cite{fm:cooperationvehicle:04} abides by a linear protocol which favors the values of more highly connected nodes.
In this way, agents which are more visible will have more of an impact on the group decision.
Yet another variation is consensus in a leader-following formation, in which a set of agents called leaders abide by individual protocol but continue to influence to rest of the network.
This problem has been studied in~\cite{rmm:controlgraphtheory:09,nc:leaderfollower:10}.

\paragraph{Laplacian linear systems}
Fast methods for solving systems of linear equations gained awareness with the nearly-linear time solver of~\cite{st:nearlylinear:04}.
Their algorithm implements a recursive procedure for sparsifiying a graph related to the coefficient matrix so that solving the system is easy at the base of the recursion.
This work was improved in~\cite{kmp:optimalsdd:10,kmp:mlognsdd:11} with a higher quality sparsifier which reduced the depth of recursion.
A parallel solver for SDD systems is given in~\cite{ps:parallelsdd:13} which runs in polylogarithmic time and nearly-linear work, an improvement to previous bounds.

The methods in this paper are closely related to previous work on approximating the discrete Green's function (or pseudo-inverse of the Laplacian matrix)~\cite{cs:hklinear:13}.
They give an algorithm for solving Laplacian linear systems with a boundary condition on a subset of vertices that improves previous time bounds by using the method of heat kernel.

\subsection{Main Results}
In the model we consider, the communication protocol followed by the agents in the network forms a linear system of equations, and the solution to the linear system is the state of the network as a function of time.
Thus, computing the consensus value involves solving a linear system.

We consider two forms of consensus.
In the group consensus framework, we seek a consensus value that is a weighted average of intial states of the system, with weights proportional to node degrees.
In this case, our algorithm computes the consensus value by approximating the state vector corresponding to the equilibrium of the system in sublinear time.
In the local framework, a subset of agents imposes an external influence on an adjacent subset.
The consensus achieved in this case is referred to as a leader-following consensus.
Our algorithm for computing leader-following consensus on the subset of followers involves sampling vectors that approximate the equilibrium state vector, and runs in sublinear time.

Specifically, our contributions are:
\begin{enumerate}
  \item
We give a new algorithm for approximating the state of a system in a weighted average consensus framework to within a multiplicative factor $(1+\epsilon)$ and additive term $O(\epsilon)$ in time $\hkprcomplexity$ where $n$ is the size of the network.
We call this algorithm \consensusalg~and present it in Section~\ref{sec:consensusapprox}.
  \item
We give a new algorithm for approximating the state of a subset of agents in a leader-following consensus framework to within a multiplicative factor $(1+\epsilon)$ and additive term $O(\epsilon)$ in time $\greenscomplexity$, where $s$ is the size of the subset of followers.
We call this algorithm \lfconsensusalg~and present it in Section~\ref{sec:leaderfollow}.
\end{enumerate}

Our sublinear time algorithms for computing consensus value rely on the efficiency of an algorithm for approximating heat kernel pagerank.
Heat kernel pagerank is introduced in detail in~\cite{chung:hkpr:07} and~\cite{chung:partitionhkpr:im09} as a variant of Personalized PageRank~\cite{bp:websearchanatomy:isdn98}.
The heat kernel pagerank approximator is introduced Section~\ref{sec:consensusapprox} and in more detail in~\cite{cs:hklinear:13}.

\section{Preliminaries}
\label{sec:prelims}

\subsection{Networked Multi-Agent Systems}
A dynamic multi-agent system is given by a tuple $G_x = (G,x)$ where $x$ is the state of the system and $G$ is the communication network topology, represented by a graph.
Namely, each agent is represented by a node and the communication network between agents is represented by the edge set $E(G)$.
Let $x_i \in \mathbb{R}$ be a real scalar value assigned to $v_i$ such that $x(t) = (x_1(t),\ldots, x_n(t))^T$ denotes the state at time $t$. 

For an undirected graph $G = (V,E)$ of size $|V| = n$, let the nodes of $V$ be arbitrarily indexed by index set $\mathcal{I} = \{1,\ldots, n\}$ such that $E \subset \mathcal{I} \times \mathcal{I}$.  
For a node $v_i\in V$, let $N_i = \{v_j \in V | (i, j) \in E\}$ be the set of \emph{neighbors} of $v_i$ and let $d_i = |N_i|$ be the \emph{degree} of $v_i$.  
Two nodes $v_i, v_j$ are said to \emph{agree} if and only if $x_i=x_j$.  The goal of consensus is to minimize the total disagreement among nodes.

\begin{definition}[Consensus]
Let the value of nodes $x$ be the solution to the equation
\begin{equation}\label{eq:autonomoussystem}
\dot{x} = f(x, u),~ x(0) \in \mathbb{R}^n.
\end{equation}
Let $\chi : \mathbb{R}^n \rightarrow \mathbb{R}$ be an operator on $x = (x_1,\ldots, x_n)^T$ that generates a decision value $\chi(x)$.  
Then we say all nodes of the graph have reached \emph{consensus with respect to $\chi$} in finite time $T > 0$ if and only if all nodes agree and $x_i(T) = \chi(x(0)) ~\forall~ i\in\mathcal{I}$.
We call $\chi(x) := \chi(x(0))$ the \emph{consensus value}.
\end{definition}

One notion of consensus is a \emph{weighted average consensus}, given by
\begin{equation*}
\chi_w(x) = \frac{\sum_i d_i x_i}{\sum_i d_i}.
\end{equation*}
We show (Theorem~\ref{thm:consheateq}) that any connected undirected graph globally asymptotically reaches weighted average consensus when each node applies the distributed linear protocol
\begin{equation}\label{eq:weightedavgprot}
u_i(t) = 1/d_i \sum\limits_{j \in N_i}(x_j(t)-x_i(t)).
\end{equation}

We assume $G$ is connected for the remainder of the paper.

\subsection{Graph Laplacians and Heat Kernel}
\label{sec:prelimlaplacian}
In this work, we consider graphs which are weight-normalized so that every entry of the weighted adjacency matrix $A$ is $a_{ij} \in \{0,1\}$, and the unordered pair $(i,j) \in E(G)$ if and only if $a_{ij} = 1$.  
Let $D$ denote the diagonal degree matrix $D(i,i) = d_i$.

The \emph{Laplacian} is defined $\L = D^{-1/2}(D-A)D^{-1/2}$.  
Let $\Delta$ be the graph matrix $\Delta = I-D^{-1}A$.
We call $\Delta$ the \emph{Laplace operator}.  
We note that $\Delta$ is similar to the matrix $\L$.\footnote{The Laplacian used in~\cite{om:consensusdynamic:03} is the matrix $L = D-A$, a common variation.} 

The \emph{heat kernel} of a graph is a solution to the heat differential equation
\[
\frac{\partial u}{\partial t} = -\Delta u.
\]
The heat kernel can be formulated in the context of random walks on graphs.  
Consider the transition probability matrix associated to a random walk given by $P = D^{-1}A$.  
Then heat kernel is defined:
\begin{align}
H_t &= e^{-t\Delta} = \sum\limits_{k=0}^{\infty}\frac{(-t)^k}{k!}\Delta^k\label{eq:hkdelta}\\
&= e^{-t(I-P)}  = \sum\limits_{k=0}^{\infty}e^{-t}\frac{t^k}{k!}P^k\label{eq:hkw}.
\end{align}

The following similarity of the heat kernel, $\H_t$, is of interest for its symmetry.
Using definition (\ref{eq:hkdelta}),
\begin{align*}
\H_t 
&= \sum_{k=0}^{\infty}\frac{t^k}{k!}D^{1/2}\Delta^kD^{-1/2}\\
&= \sum_{k=0}^{\infty}\frac{t^k}{k!}D^{1/2}(D^{-1/2}\L D^{1/2})^kD^{-1/2}\\
&= \sum_{k=0}^{\infty}\frac{t^k}{k!}\L^k
= e^{-t\L}.
\end{align*}

Heat kernel pagerank is a row vector determined by two parameters; $t\in\R^+$, and a preference row vector $f\in\R^n$. 
It is given by the following equation:
\[
\hkpr = fH_t = \sum\limits_{k=0}^{\infty} e^{-t}\frac{t^k}{k!}fP^k.
\]

Specifically, it is an exponential sum of random walks generated from a starting vector, $f\in \mathbb{R}^n$.  
As an added benefit, heat kernel pagerank simultaneously satisfies the heat equation with the rate of diffusion controlled by the parameter $t$.  
Both properties are powerful tools in consensus problems.

\section{Heat Kernel Pagerank for Weighted-Average Consensus}
\label{sec:consensus}
In this section we present a linear consensus protocol for a dynamic network and show how to compute a weighted average consensus for the protocol using heat kernel pagerank.

We first recall some principles of control theory.
Consider the system with controls as in (\ref{eq:autonomoussystem}).
A point $x_e$ is an \emph{equilibrium point} of the system if $f(x_e,u) = 0$, and $x_e$ is an equilibrium point if and only if $x(t) = x_e$ is a trajectory.
The system is \emph{globally asymptotically stable} if, for every trajectory $x(t)$, $x(t) \rightarrow x_e$ as $t\rightarrow\infty$.
To check this, two conditions are sufficient.

\begin{definition}[Global asymptotic stability]\label{def:gas}
A system is \emph{globally asymptotically stable} if
\begin{enumerate}
\item it is stable in the Lyapunov sense, and\label{condition:lyapunov}
\item the equilibrium $x_e$ is convergent, i.e., for every $\epsilon > 0$, there is some time $T$ such that 
\[
\norm{x(0)-x_e} < \delta ~~~\mbox{ means }~~~ \norm{x(t)-x_e} < \epsilon
\]
for every time $t>T$.\label{condition:convergent}
\end{enumerate}
\end{definition}
In particular, when considering a time-invariant linear state space model $\dot{x} = -Mx$, for some matrix $M$, condition \ref{condition:lyapunov} is satisfied if $M$ is positive semidefinite.

\subsection{Consensus and the Laplacian}
Consider the network of integrator agents with dynamics $\dot{x}_i = u_i$ where each agent applies the distributed linear protocol (\ref{eq:weightedavgprot}).
We can characterize the dynamics of the system by the Laplace operator for the underlying graph, as described by the following theorem:

\begin{theorem}
\label{thm:consheateq}
Let $G_x$ be a dynamic multi-agent system and suppose each node of $G$ applies the distributed linear protocol (\ref{eq:weightedavgprot}).
Then the value of $x$ at time $t$ is given by the solution to the system
\begin{align}
\dot{x}(t) = -\Delta x(t), ~~x(0)\in \mathbb{R}^n. \label{eq:consensus}
\end{align}
Additionally, this protocol globally asymptotically reaches a weighted average consensus.
\end{theorem}

\begin{proof}
Let $x_e$ be an equilibrium of the system $\dot{x} = -\Delta x$.
Then by definition of equilibrium, $\Delta x_e = 0$ and therefore $x_e$ is a right eigenvector associated to the eigenvalue $\lambda=0$.
In particular $x_e$ is in the null space of $\Delta$.
Since $G$ is connected, $\Delta$ has exactly one zero eigenvalue.
Upon consideration, we see that the corresponding eigenvector is $\mathbf{1}$, the all-one's vector, as the row sums of $\Delta$ are all exactly zero.
Thus, $x_e = \alpha\mathbf{1}$ for some $\alpha\in\mathbb{R}$.
Now, note that $\sum_i u_i = 0$ for the protocol (\ref{eq:weightedavgprot}), and so the weighted average value $\chi_w(x(t))$, determined by $u(t)$, is in fact invariant with respect to $t$.
In other words, $\chi_w(x(0)) = \chi_w(x_e)$, and
\[
\chi_w(x_e) = \frac{\sum_i d_i (x_e)_i}{\sum_i d_i} = \alpha.
\]
Therefore this equilibrium is in fact the weighted average of the initial values of the nodes, and all nodes reach this value.
Also, as the system is time-invariant, the system is stable since $\Delta$ is positive semidefinite.

By Definition~{\ref{def:gas}}, the Theorem is proved.
\end{proof}

Now we can summarize the state of the system with a single heat kernel pagerank vector.

\begin{theorem}\label{thm:maintruesolution}
Let $G_x$ be a dynamic multi-agent system and suppose each node of $G$ applies the distributed linear protocol (\ref{eq:weightedavgprot}).
Let $D$ be the diagonal degree matrix of $G$.
Then the state of the system is given by
\begin{equation}\label{eq:weightedavghkpr}
x(t) = \rho_{T, f}D^{-1}, ~~~f=x(0)^{tr}D,
\end{equation}
where $M^{tr}$ denotes the transpose.
\end{theorem}

\begin{proof}
The solution to (\ref{eq:consensus}) is the evolving state of the system.
This solution is
\begin{equation}\label{eq:consvalueexp}
x(t) = e^{-t\Delta}x(0) = H_tx(0).
\end{equation}
Using the symmetrized version of heat kernel,
\begin{align}
x(t) &= (D^{-1/2}\H_t D^{1/2})x(0)\nonumber\\
x(t)^{tr} &= x(0)^{tr}D^{1/2}(D^{1/2}H_t D^{-1/2})D^{-1/2}\label{line:transpose}\\
x(t)^{tr} &= (x(0)^{tr}D)H_tD^{-1}\nonumber,
\end{align}
where line~\ref{line:transpose} uses the symmetry of $D$ and $\H$.
Thus, the values $x(t)$ given by (\ref{eq:consvalueexp}) are related to the heat kernel pagerank vector $\hkpr$ with preference vector $f = x(0)^{tr}D$.
\end{proof}

To compute the equilibrium state at which all agents reach consensus, we know that time $T = O(1/\lambda_1)$ is an upper bound.
Figure~\ref{fig:tconverge} depicts the results of computing weighted average consensus with heat kernel pagerank as in Theorem~\ref{thm:maintruesolution} with different values for $t$.
The network is an undirected social network of dolphins~\cite{dolphins} with initial state values randomly chosen from the interval $(0,1)$ (Figure~\ref{fig:dolphins}).
The chart plots total disagreement $\norm{\delta}$ for disagreement vector $\delta(t) = x(t) - \chi_w(x)\mathbf{1}$, where $x(t) = \hkpr D^{-1}$ for $f=x(0)^{tr}D$.
The vertical line corresponds to $t=1/\lambda_1$.

\begin{figure*}
\centering
\begin{subfigure}{.5\textwidth}
  \centering
  \includegraphics[width=0.9\linewidth]{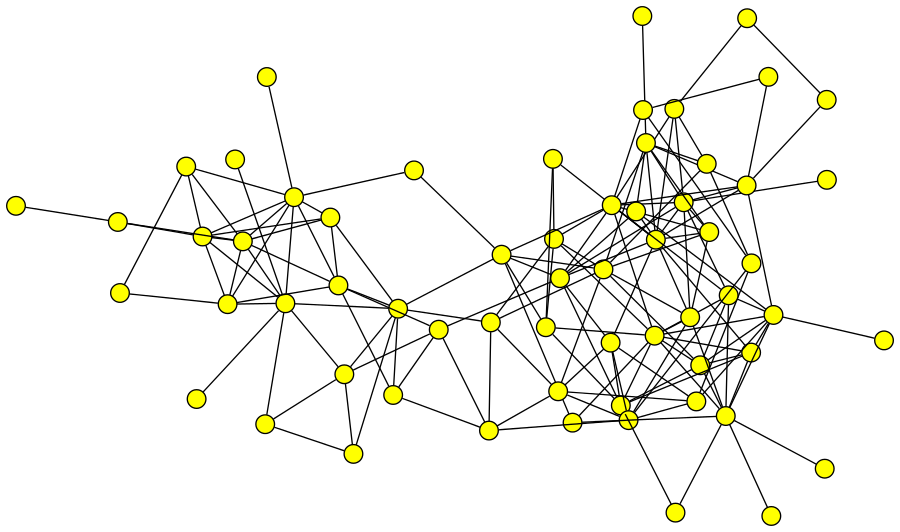}
  \caption{Dolphin social network.}
  \label{fig:dolphins}
\end{subfigure}%
\begin{subfigure}{.5\textwidth}
  \centering
  \includegraphics[width=1.0\linewidth]{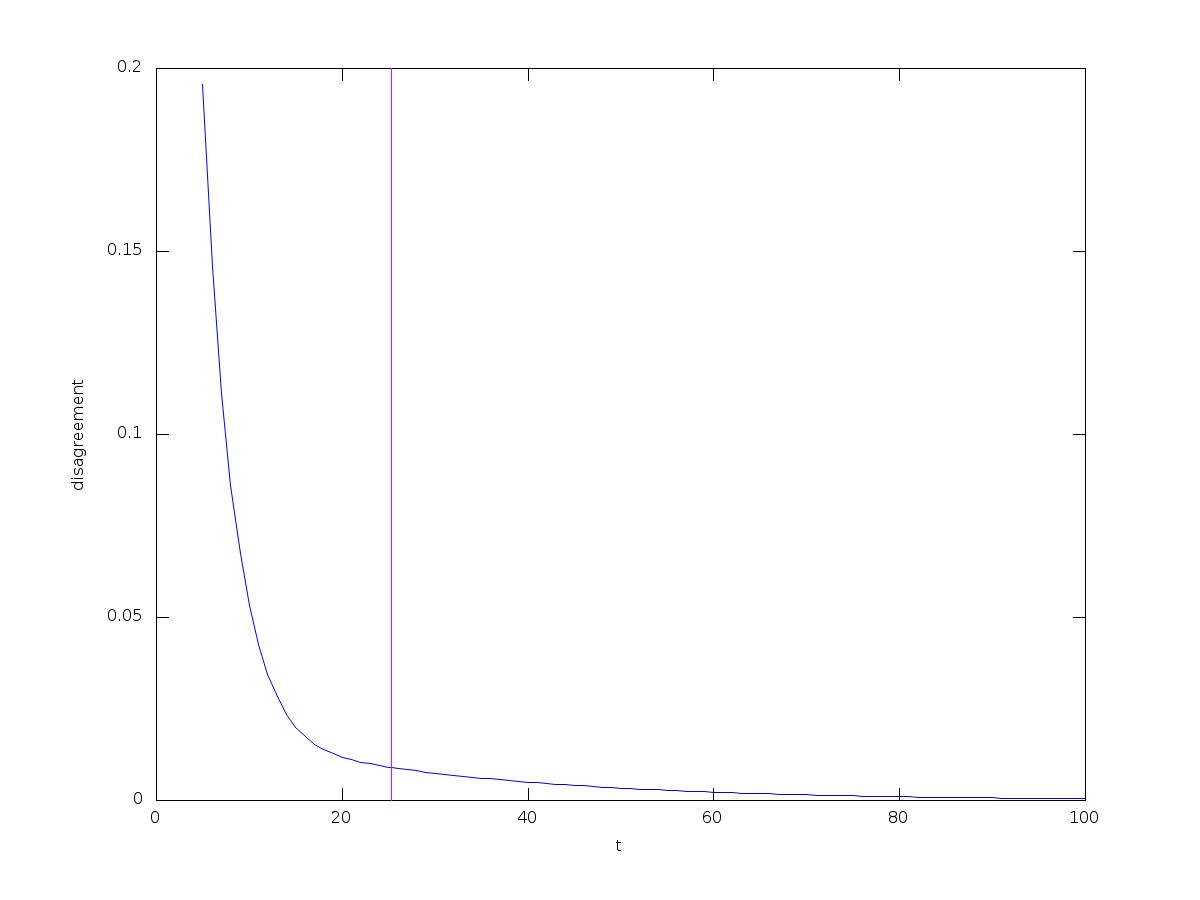}
  \caption{Total disagreement over varying times $t$.  Disagreement is computed in terms of the weighted average consensus $\chi_w(x_0)$.  The red line denotes the data point for $t=1/\lambda_1$.}
  \label{fig:tconverge}
\end{subfigure}
\caption{Weighted average consensus convergence results.}
\end{figure*}

\subsection{An Algorithm for Computing Consensus Value Using Approximate Heat Kernel Pagerank}
\label{sec:consensusapprox}
Our weighted average consensus algorithm uses an algorithm for approximating heat kernel pagerank as a subroutine.  
We use the following definition of an approximate heat kernel pagerank.

\begin{definition}
Let $f\in\R^n$ be a vector over nodes of a graph $G = (V,E)$ and let $\hkpr$ be the heat kernel pagerank vector over $G$ according to $t$ and $f$.  
Then we say that $\nu \in \mathbb{R}^n$ is an \emph{$\epsilon$-approximate heat kernel pagerank vector} if
\begin{enumerate}
\item for every node $v_i \in V$ in the support of $\nu$, \newline$(1-\epsilon)\hkpr[i] -\epsilon \leq \nu[i] \leq (1+\epsilon)\hkpr[i]$, and
\item for every node with $\nu[i] = 0$, it must be that $\hkpr[i] \leq \epsilon$.
\end{enumerate}
\end{definition}

\begin{theorem}[Weighted Average Consensus in Sublinear Time]\label{thm:hkconsensus}
Let $G_x$ be a dynamic $n$-agent system and suppose each node of $G$ applies the distributed linear protocol (\ref{eq:weightedavgprot}).
Then the state of the system can be approximated to within a multiplicative factor of $(1+\epsilon)$ and an additive term of $O(\epsilon)$ for any $0<\epsilon<1$ in time $\hkprcomplexity$. 
\end{theorem}

We call the algorithm \consensusalg.
The algorithm makes a call to \hkpralg, an extension of the algorithm presented in~\cite{cs:hklinear:13} for quickly computing an approximation of a restricted heat kernel pagerank vector.
For the sake of completeness, the algorithm and a summary of the results of~\cite{cs:hklinear:13} are given at the end of this section.

\begin{proof}[Proof of Theorem~\ref{thm:hkconsensus}]
First, the $\epsilon$-approximate vector $x$ returned by \hkpralg~ is an approximation of the true state by Theorem~\ref{thm:maintruesolution}.
Thus we have left to verify the approximation guarantee and the running time.
The total running time is dominated by the heat kernel pagerank approximation, which is $\hkprcomplexity$ by Theorem~\ref{thm:approxhkprsummary}, below.
Theorem~\ref{thm:approxhkprsummary} also verifies the approximation guarantee.
\end{proof}

\begin{center}
\begin{pseudocode}[framebox]{\consensusalg}{G,x,t,\epsilon}
\COMMENT{\underline{input}:}\\
\COMMENT{~~G as the $(0,1)$-adjacency matrix}\\
\COMMENT{~~$x$, initial state vector}\\
\COMMENT{~~$t\in\R^+$}\\
\COMMENT{~~$0<\epsilon<1$, error parameter}\\
\COMMENT{\underline{output}: $x(t)$}\\
D \GETS \mbox{diagonal matrix of rowsums(G)}\\
f \GETS x^TD\\
y \GETS \hkpralgparams\\
\RETURN {yD^{-1}}
\end{pseudocode}
\end{center}

\subsection{A Sublinear Time Heat Kernel Pagerank Approximation Algorithm}
The analysis for approximating heat kernel pagerank follows easily from that for a restricted heat kernel pagerank vector by considering the entire vertex set rather than a subset.
We refer the reader to~\cite{cs:hklinear:13} for a more complete description.

\begin{algorithm}[H]
\floatname{algorithm}{Algorithm}
\caption{\hkpralgparams}
\label{}
\algblock[Name]{Start}{End}
input: a graph $G$, $t\in\R^+$, preference vector $f\in\R^n$, error parameter $0 < \epsilon < 1$.\\
output: $\rho$, an $\epsilon$-approximation of $\hkpr$.\\
\begin{algorithmic}
  \State initialize $0$-vector $\rho$ of dimension $n$, where $n=|V|$
  \State $r \gets \rparam$
  \State $K \gets \kparam$
  \State $f' \gets f/\onenorm{f}$
  \Comment{\textit{normalize $f$ to be a probability distribution vector}}
  \For {$r$ iterations}
    \State choose a starting vertex $u$ according to the distribution vector $f'$
    \Start
      \State simulate a $P$ random walk where $k$ steps are taken with probability $e^{-t}\frac{t^k}{k!}$ and $k\leq K$
      \State let $v$ be the last vertex visited in the walk
      \State $\rho[v] \gets \rho[v] + 1$
    \End
  \EndFor\\
  \State $\rho \gets 1/r \cdot \rho$\\
  \Return $\onenorm{f}\cdot\rho$
\end{algorithmic} 
\end{algorithm}

\begin{theorem}\label{thm:approxhkprsummary}
Let $G$ be a graph, $t\in\R^+$, and $f\in\R^n$.
Then, the algorithm \hkpralgparams outputs an $\epsilon$-approximate vector $\hkapprox$ of the heat kernel pagerank $\hkpr$ for $0 < \epsilon < 1$ with probability at least $1-\epsilon$.
The running time of \hkpralg~is $\hkprcomplexity$.
\end{theorem}

\section{Heat Kernel Pagerank for Consensus in Leader-Following Formations}
\label{sec:leaderfollow}
In this section we consider a multi-agent network in which a certain subset of agents $l\subset V$ are \emph{leaders}, and the rest $f = V\setminus l$ are dubbed \emph{followers}.
In this scenario, leaders will adjust their values according to individual protocol, while followers in the system adjust according to communication channels as usual.
The consensus goal in this case is a \emph{leader-following consensus}, in which all agents agree on a value by following the leaders.

Let $u^f$ denote the protocol among the set of followers and let $u^l$ denote the control dictated by the leaders and influencing the followers.
Similarly, let $x^f$ denote the state of the followers and $x^l$ denote the state of the leaders.
The vectors $x^f$ and $x^l$ can be understood as the usual state vector $x$ restricted to following and leading agents, respectively.
Then we have the following definition.

\begin{definition}[Leader-following consensus]
A \emph{leader-following consensus} of a system is achieved if for every agent $v_i$ there is a local protocol $u_i$ 
such that $x_i(T) = \chi_{lf}(x(0))$ for some finite time $T>0$ and some operator $\chi_{lf}:\R^n \rightarrow \R$.
In this case, we call the value $\chi_{lf}(x(0))$ the leader-following consensus value.
\end{definition}

For the protocol
\begin{equation}\label{eq:normalizedprot}
u_i(t) = 1/d_i \sum\limits_{j\in N_i} \Bigg(\sqrt{\frac{d_i}{d_j}}x_j(t) - x_i(t)\Bigg),
\end{equation}
the value for $x$ is given by the dynamics
\[
u(t) = \dot{x}(t) = -\L x(t).
\]
We let the followers abide by protocol (\ref{eq:normalizedprot}).

Let $\L_f$ be the Laplacian $\L$ restricted to rows and columns corresponding to the followers, and $\L_{fl}$ be $\L$ with rows restricted to the followers and columns restricted to the leaders.
Then the dynamics of the followers can be summarized by:
\[
\dot{x}^f(t) = -\L_f x^f(t) - \L_{fl}u^l(t).
\]
Since $\dot{x}^f$ is control of the subnetwork induced by the group of followers, this can be rewritten
\begin{align*}
u^f &= -\L_f x^f - \L_{fl}u^l~~~\mbox{ or alternatively, }\\
x^f &= -\L_f^{-1}u^f + \L_f^{-1}\L_{fl}u^l.
\end{align*}

Indeed, as long as the subgraph induced by the subset of followers is connected, the inverse $\L_f^{-1}$ exists.
We have arrived at the following.

\begin{theorem}
Let $G_x$ be a dynamic multi-agent system with proper subsets of leaders, $l\subset V$, and followers, $f = V\setminus l$, such that the induced subgraph on $f$ is connected.
Suppose the followers apply the protocol (\ref{eq:normalizedprot}), and suppose the leaders apply some individual protocol $u_i = f(x_i)$ dictated only by that leader's state.
Then the followers' state values $x^f$ at time $t$ are given by the solution to the system
\[
\L_f x^f(t) = b(t), ~~~\mbox{ where }~~~ b(t)=-(u^f(t) + \L_{fl}u^l(t)).
\]
\end{theorem}

An efficient algorithm called \greensalg~ for solving linear systems $\L_f x^f = b$ with a linear protocol applied to a subset specified by $b$ is given in~\cite{cs:hklinear:13}.
They show that the solution $x_f$ can be computed with the symmetric heat kernel using the relationship
\begin{equation}\label{eq:inversegreen}
\L_f^{-1}b = \int_{0}^{\infty}(\H_t)_f ~\mathrm{d}t~ b,
\end{equation}
where $(\H_t)_f$ is $\H_t$ with rows and columns restricted to the set $f$.

The solution $\L_f^{-1}b$ can be approximated by sampling sufficiently many values of $(\H_t)_f b(t)$.
Further, it is given that the solution can be approximated in $\greenscomplexity$ time, where $s$ is the size of the subset of followers.

\begin{center}
\begin{pseudocode}[framebox]{\lfconsensusalg}{G,x,t,f,l,u^l,\epsilon}
\label{alg:lfcons}
\COMMENT{\underline{input}:}\\
\COMMENT{~~G as the $(0,1)$}\\
\COMMENT{~~$x$, initial state vector}\\
\COMMENT{~~$t\in\R^+$}\\
\COMMENT{~~$f$, subset of followers}\\
\COMMENT{~~$l$, subset of leaders}\\
\COMMENT{~~$u^l$, protocol applied by the leaders}\\
\COMMENT{~~$0<\epsilon<1$, error parameter}\\
\COMMENT{\underline{output}: $x(t)$}\\
\PROCEDURE{FollowerProt}{G,x,t}
\FOREACH i\in l \DO
    u^f[i] \GETS u_i(t) = 1/d_i \sum\limits_{j\in N_i} \Big(\sqrt{\frac{d_i}{d_j}}x_j(t) - x_i(t)\Big)\\
\RETURN{u^f}
\ENDPROCEDURE
\PROCEDURE{b}{t}
    u^f \GETS \CALL{FollowerProt}{G,x,t}\\ 
    b \GETS -(u^f(t) + \L_{fl}u^l)\\
\ENDPROCEDURE
b \GETS \CALL{b}{t}\\
s \GETS |f|\\
T \GETS \tparamsolver\\
N \GETS T/\epsilon\\
r \GETS \rparamsolver\\
\mbox{initialize a }0-\mbox{vector }x^f\mbox{ of dimension }s\\
\FOR i=0 \TO r \DO
    \BEGIN
    \mbox{draw $j$ from $[1,N]$ uniformly at random}\\
    x_i \GETS \texttt{ApproxHK}(G,jT/N,b,f,\epsilon)\\
    x^f \GETS x^f + x_i\\
    \END\\
\RETURN{\frac{1}{r}xD_S^{-1/2}}
\end{pseudocode}
\end{center}

The running time and approximation guarantees of \lfconsensusalg~follow from the running time of \greensalg~\cite{cs:hklinear:13}, and we have the following:

\begin{theorem}[Leader-Following Consensus in Sublinear Time]
\label{thm:lfconsensus}
Let $G_x$ be a dynamic multi-agent system with proper subsets of leaders, $l\subset V$, and followers, $f = V\setminus l$, such that the induced subgraph on $f$ is connected.
Suppose the followers apply the protocol (\ref{eq:normalizedprot}), and suppose the leaders apply some individual protocol $u_i = f(x_i)$ dictated only by that leader's state.
Then the state of the system can be approximated to within a multiplicative factor of $(1+\epsilon)$ and an additive term of $O(\epsilon)$ for any $0<\epsilon<1$ in time $\greenscomplexity$, where $s$ is the size of the subset of followers.
\end{theorem}

\section{Discussion}\label{sec:discussion}
The significance of sublinear running times is scalability.
The robustness and efficiency of the algorithms \consensusalg~and \lfconsensusalg~are of great importance for networks too large to fit in memory, and the running time/approximation tradeoff allows for appropriate tuning.
This is especially notable for local algorithms, which reduce computation over large networks to a small subset.
For instance, while the group of leaders may be small in a leader-following framework, the difference in complexity for computing consensus in a leader-following formation as opposed to full group consensus can be significant.
The subset of followers influenced by the leaders may be a small portion of the entire graph, so that $s<<n$, and we are spared work over the entire graph in the case that we are interested in only a small area.
In these cases, the gain in running times are valueable.


\addtolength{\textheight}{-12cm}   



%
%
%
%


{\footnotesize
\bibliographystyle{abbrv}
\bibliography{consensus}
}

\end{document}